\documentclass[sigplan,screen]{acmart}

\usepackage{tikz} \usetikzlibrary{decorations.pathreplacing}
\usepackage{amsthm}
\usepackage{stmaryrd}
\usepackage{dashbox}
\usepackage{commath}
\usepackage{ebproof}

\newtheorem{theorem}{Theorem}
\newtheorem{proposition}[theorem]{Proposition}

\newtheorem{lemma}[theorem]{Lemma}
\newtheorem{corollary}[theorem]{Corollary}
\theoremstyle{definition}
\newtheorem{definition}[theorem]{Definition}

\newtheorem{remark}[theorem]{Remark}

\newcommand\dashedph[1][H]{\setlength{\fboxsep}{0pt}\setlength{\dashlength}{2.2pt}\setlength{\dashdash}{1.1pt} \dbox{\phantom{#1}}}

\renewcommand{\lor}{\vee}
\renewcommand{\land}{\wedge}
\renewcommand{\lnot}{\neg}
\newcommand{\limp}{\rightarrow}

\newcommand{\bor}{\vee}
\newcommand{\band}{\wedge}
\newcommand{\bnot}{\neg}
\newcommand{\cc}{{\operatorname{c\kern-.05emc}}}
\renewcommand{\sc}{\mathbin{\cdot}}
\newcommand{\st}{\mathbin{\star}}
\newcommand{\cln}{\,\colon}
\newcommand{\cldotsc}{, \allowbreak \ldots, \allowbreak }
\newcommand{\multi}[1]{\overline{#1}}
\newcommand\pole{{\protect\mathpalette{\protect\polehelper}{\bot}}} \def\polehelper#1#2{\mathrel{\rlap{$#1#2$}\mkern3mu{#1#2}}}
\newcommand{\falsity}[1]{\left\lVert#1\right\rVert}
\newcommand{\truth}[1]{\left\lvert#1\right\rvert}
\newcommand{\instrU}{\zeta}
\newcommand{\instrR}{\eta}
\renewcommand{\th}[1]{\operatorname{Th}(#1)}

\newcommand{\ds}[1]{\llbracket #1 \rrbracket}

\copyrightyear{2022}
\acmYear{2022}
\setcopyright{licensedothergov}\acmConference[LICS '22]{37th Annual ACM/IEEE Symposium on Logic in Computer Science (LICS)}{August 2--5, 2022}{Haifa, Israel}
\acmBooktitle{37th Annual ACM/IEEE Symposium on Logic in Computer Science (LICS) (LICS '22), August 2--5, 2022, Haifa, Israel}
\acmPrice{15.00}
\acmDOI{10.1145/3531130.3532484}
\acmISBN{978-1-4503-9351-5/22/08}

\begin{document}
%
\title{A first-order completeness result about characteristic Boolean algebras in classical realizability}


\author{Guillaume Geoffroy}
\email{guillaume.geoffroy@irif.fr}
\affiliation{%
  \institution{Université Paris Cité, IRIF}
  \city{Paris}
  \country{France}
}

\begin{abstract}
We prove the following completeness result about classical realizability: given any Boolean algebra with at least two elements, there exists a Krivine-style classical realizability model whose characteristic Boolean algebra is elementarily equivalent to it. This is done by controlling precisely which combinations of so-called ``angelic'' (or ``may'') and ``demonic'' (or ``must'') nondeterminism exist in the underlying model of computation.
\end{abstract}

\maketitle

\section{Introduction}

\paragraph*{Classical realizability} Realizability is an aspect of the propositions-as-types / proofs-as-programs correspondence in which each proposition is interpreted as a specification on the behaviour of programs: programs which satisfy this specification are said to \emph{realize} the proposition. This interpretation defines a notion of truth value: a proposition counts as ``true'' if it is realized by a well-formed program. In particular, any provable proposition is true in that sense. Indeed, any proof, when be seen as a program through the correspondence, must realize the proposition that it proves. This fundamental result ensures that realizability is compatible with logical deduction.

Initially, this compatibility was restricted to intuitionistic deduction. Griffin's discovery of a link between control operators and classical reasoning \cite{griffin:callcc} overcame this limitation. More precisely, Griffin proved that Peirce's law--a deductive principle that is valid in classical logic but not in intuitionistic logic--can be used as a specification (i.e. as a type) for Scheme's operator \emph{call/cc} (``call with current continuation''), which allows a program to manipulate its own evaluation context as a first-class object.

Using this idea, Krivine developed a framework which could interpret all classical reasoning, first within second-order arithmetic \cite{krivine:depchoice}, and then within Zermelo–-Fr\ae{}nkel set theory with dependent choice \cite{krivine:ra1}. Miquel then adapted this \emph{classical realizability} to higher-order arithmetic and explored its connections with forcing \cite{miquel:lics11}. Work on interpreting reasoning that uses the full axiom of choice is ongoing \cite{krivine:model-full-choice}.

\paragraph*{Characteristic Boolean algebras}

Abstractly, a \emph{classical realizability model} is the data of a \emph{model of computation} (for example: a variant of the lambda-calculus enriched with the instruction \emph{call/cc}), a \emph{model of deduction} (for example: a first-order language, plus the rules of classical reasoning, and optionally a theory on this language, i.e. a set of axioms), and a \emph{realizability relation} between \emph{programs} (from the former) and \emph{propositions} from the latter.

Each classical realizability model contains a \emph{characteristic Boolean algebra} (which Krivine calls $\gimel 2$--``gimel $2$''). More precisely, each formula $A$ in the language of Boolean algebras can be translated into a proposition which is usually denoted by $\gimel 2 \models A$ -- read ``the characteristic Boolean algebra satisfies $A$''.

In any given classical realizability model, the set of all first-order formulas $A$ such that the proposition ``the characteristic Boolean algebra satisfies $A$'' is realized by a well-formed program (i.e. ``true'') forms a first-order theory on the language of Boolean algebras: this is called the \emph{first-order theory of the characteristic Boolean algebra of the realizability model}. This theory may or may not be consistent, but it is always closed under classical deduction, and it always contains the theory of Boolean algebras with at least two elements.

The characteristic Boolean algebra, and in particular the ability to ``shape'' it, plays a central role in classical realizability. For example, consider Krivine's \emph{model of threads} \cite{krivine:ra2}. One of its remarkable combinatorial properties is that in this model, there is a whole atomless Boolean algebra embedded in the poset of the cardinalities between the countable and the continuum; and the way this property was obtained was by first making the characteristic Boolean algebra is atomless, and then embedding it in this poset. As an other example, Krivine's construction of a particular classical realizability model that satisfies the axiom of choice \cite{krivine:model-full-choice} depends crucially on the ability to reliably force a realizability model's characteristic Boolean algebra to be isomorphic to any given \emph{finite} Boolean algebra with at least $2$ elements (in that case, the Boolean algebra with $4$ elements).

\paragraph*{Contribution}

The contribution of this paper is to prove that the characteristic Boolean algebra can in fact be made elementarily equivalent to \emph{any} given Boolean algebra with at least two elements, finite or not. More precisely, we prove that for each first-order theory $\mathcal T$ over the language of Boolean algebras, the following two conditions are equivalent:
\begin{itemize}
\item The theory $\mathcal T$ is closed under classical deduction and contains the theory of Boolean algebras with at least two elements;
\item There exists a classical realizability model whose characteristic Boolean algebra's theory is exactly $\mathcal T$.
\end{itemize}
Note that this fact holds independently from the consistency of the theory $\mathcal T$.

In particular, given a first-order formula $A$ over the language of Boolean algebras, the proposition ``the characteristic Boolean algebra satisfies $A$'' is universally realized (i.e. realized in all models) if and only if $A$ is true in all Boolean algebras with at least two elements.

The proof we give is constructive: given a theory $\mathcal T$, we describe a concrete realizability model whose characteristic Boolean algebra's theory is $\mathcal T$. The construction works as follows: it has been pointed out \cite{geoffroy:nondeterminism} that the properties of the characteristic Boolean algebra reflect the kinds of nondeterminism that exist in the underlying computational model; so for each formula $A$ in $\mathcal T$, what we do is add to the computational model a nondeterministic instruction $\gamma_A$ which has exactly the right combination of so-called ``angelic'' (or ``may'') and ``demonic'' (or ``must'') nondeterminism to realize the proposition ``the characteristic Boolean algebra satisfies $A$''.

\paragraph*{Outline}

Section \ref{section:conventions} states well-known facts about classical realizability (including the fact that in the equivalence we want to prove, the second condition implies the first), and lays down the conventions that will be used throughout the paper. To keep the discussion focused and the notations simple, we restrict the language of propositions to the first-order language of Boolean algebras (rather than, say, the language of set theory, or the second-order language of Peano arithmetic). The main benefit is that, in this context, given any first-order formula $A$ in the language of Boolean algebras, the proposition ``the characteristic Boolean algebra satisfies $A$'' is simply the formula $A$: no translation is needed.

Section \ref{section:construction} details the construction, given any first-order theory $\mathcal T$ that is closed under classical deduction and contains the theory of Boolean algebras with at least two elements, of a realizability model that satisfies $\mathcal T$.

Section \ref{section:contains-T} proves that this model's characteristic Boolean algebra's theory does indeed contain $\mathcal T$, and Section \ref{section:contained-in-T} proves the converse inclusion, which concludes the proof of this paper's main result.

Finally, Section \ref{section:sequential} gives an example of application of this result to the problem of sequentialisation in a denotational model of the lambda-calculus with a control operator.

\section{Conventions and reminders about classical realizability}\label{section:conventions}

\subsection{First-order formulas on Boolean algebras}

The \emph{language of Boolean algebras} is the first-order language with equality over the signature $(0,1,\bor,\band,\bnot)$ (respectively: two constants, two binary function symbols with infix notation, and one unary function symbol). To make it clear which symbols we take as primitives, we spell out its grammar:

First-order terms:
\[ \arraycolsep=2pt \begin{array}{rlll}
a,b & := &       & z \text{ (first-order variable)}\\
    &    & \vert & 0 ~\vert~ 1 ~\vert~ a \lor b ~\vert~ a \land b ~\vert~ \lnot a
\end{array} \]

First-order formulas:
\[ \begin{array}{rll}
A,B & := & a \neq b ~\vert~ A \limp B ~\vert~ \forall z\, A
\end{array} \]

Note that, as is customary in classical realizablity, we take non-equality rather than equality as a primitive symbol, because its realizability interpretation is simpler.

The other usual symbols can be encoded as follows:
\begin{itemize}
\item $\bot$ is $0 \neq 0$,
\item for all first-order terms $a,b$, $a = b$ is $(a \neq b) \limp \bot$,
\item for all first-order formulas $A,B$, $A \wedge B$ is $(A \limp B \limp \bot) \limp \bot$,
\item for all first-order formulas $A,B$, $A \vee B$ is $(A \limp \bot) \limp (B \limp \bot) \limp \bot$,
\item for all first-order formulas $A$ and all first-order variables $z$, $\exists z~ A$ is $(\forall z~ (A \limp \bot)) \limp \bot$.
\end{itemize}

A set of closed first-order formulas is called a \emph{first-order theory}. 
Over the signature we have chosen, the theory of Boolean algebras can be axiomatised by a finite set of equations. As a result, there exists a finite first-order theory $\mathcal{T}_{\operatorname{Bool}}$ consisting of:
\begin{itemize}
\item the first-order formula $0 \neq 1$,
\item plus a finite number of closed first-order formulas of the form $\forall \multi{z}~ a = b$ (where $\multi{z}$ is a list of variables),
\end{itemize}
such that for each first-order structure $\mathbb B$ over the language of Boolean algebras, $\mathbb B$ satisfies $\mathcal{T}_{\operatorname{Bool}}$ if and only if $\mathbb B$ is a Boolean algebra with at least two elements.

First-order formulas are defined up to $\alpha$-renaming. Given a first-order formula $A$ (repectively, a first-order term $a$), a list $\multi{z}$ of variables and a list $\multi{b}$ of first-order terms of equal length, we denote by $A[\multi{z} := \multi{b}]$ (respectively, $a[\multi{z} := \multi{b}]$) the simultaneous, capture-avoiding substitution of $\multi{z}$ with $\multi{b}$ in $A$ (respectively, in $a$).

\subsection{The $\lambda_c$-calculus}

\paragraph*{Syntax} The $\lambda_c$-calculus consists of three kinds of syntactic objects: \emph{$\lambda_c$-terms} (which represent programs), \emph{stacks} (which represent execution environments), and processes (which represent a program running in a given environment). They are defined by the following grammars, up to $\alpha$-renaming:

$\lambda_c$-terms:
\[ \arraycolsep=2pt \begin{array}{rlll}
t,u & := & & x ~\vert~ t u ~\vert~ \lambda x.\, t \\
&& \vert & \cc \text{\quad(call with current continuation)} \\
&& \vert & k_\pi \text{\quad($\pi$ stack)} \\
&& \vert & \instrU_n \text{\quad($n\in \mathbb N$: unrestricted additional instructions)} \\
&& \vert & \instrR_n \text{\quad($n\in \mathbb N$: restricted additional instructions)} \\
\end{array} \]

Stacks:
\[ \begin{array}{rlll}
\pi,\pi' & := & & t \sc \pi \text{\quad($t$ \emph{closed} $\lambda_c$-term)}\\
&& \vert & \omega \text{\quad(empty stack)}
\end{array} \]

Processes:
\[ \begin{array}{rll}
p,q & := & t \st \pi \text{\quad($t$ \emph{closed} $\lambda_c$-term)}
\end{array} \]

Given a $\lambda_c$-term $t$, a list $\multi{x}$ of variables and a list $\multi{u}$ of $\lambda_c$-terms of equal length, we denote by $t[\multi{x} := \multi{u}]$ the simultaneous, capture-avoiding substitution of $\multi{x}$ with $\multi{u}$ in $t$.

\paragraph*{Operational semantics} Processes are evaluated according to the rules of the \emph{Krivine abstract machine}. Namely, we denote by $\succ_1$ (``evaluates in one step to'') the least binary relation on the set of processes such that:
\[ \begin{array}{rcll}
t u \st \pi & \succ_1 & t \st u \sc \pi & \text{(Push)}\\
\lambda x.\, v \st t \sc \pi & \succ_1 & v[x := t] \st \pi & \text{(Grab)}\\
\cc \st t \sc \pi & \succ_1 & t \st k_\pi \sc \pi & \text{(Save)}\\
k_{\pi_2} \st t \sc {\pi_1} & \succ_1 & t \st \pi_2 & \text{(Restore)}\\
\end{array} \]
for all closed terms $t,u, \lambda x.\, v$ and all stacks $\pi, \pi_1, \pi_2$.

Moreover, we denote by $\succ$ (``evaluates to'') the reflexive and transitive closure of $\succ_1$.

The rules Push and Grab simulate weak head $\beta$-reduction, and the rules Save and Restore allow programs to manipulate continuations ($\cc$ stands for ``call with current continuation''). In the context of realizability, the former pair will ensure compatibility with intuitionistic logic, and the latter with classical logic.

Note that there are no rules for the additional instructions: their purpose will be to help construct specific \emph{poles} (see next subsection), and they can be ignored for the time being.

\paragraph*{Typing}

Typing judgements have the following form: $x_1 : A_1\cldotsc x_n \cln A_n \vdash t \cln B$, where $x_1\cldotsc x_n$ are pairwise distinct variables, $t$ is a $\lambda_c$-term with no free variables other than $x_1\cldotsc x_n$, and $A_1\cldotsc A_n$ are first-order formulas (possibly with free variables).

Typing judgements are defined up to $\alpha$-renaming (i.e. $x_1 \cln A_1\cldotsc x_n \cln A_n \vdash t \cln B$ is the same as $y_1 \cln A_1\cldotsc y_n \cln A_n \vdash t[\multi{x} := \multi{y}] : B$), and up to permutations of the context (i.e. $x_1 \cln A_1\cldotsc x_n \cln A_n \vdash t \cln B$ is the same as $x_{\sigma(1)} \cln A_{\sigma(1)}\cldotsc x_{\sigma (n)} \cln A_{\sigma(n)} \vdash t \cln B$ for all permutations $\sigma$).

A typing judgement is \emph{valid} if it can be derived from the following rules:

\smallskip

$\hfill
\begin{prooftree}
  \hypo{\Gamma, x \cln A \vdash t \cln B }
  \infer1{\Gamma \vdash \lambda x.\, t \cln A \limp B}
\end{prooftree}
\hfill
\hfill
\begin{prooftree}
  \hypo{\Gamma \vdash t \cln A \limp B }
  \hypo{\Gamma \vdash u \cln A }
  \infer2{\Gamma \vdash tu \cln B}
\end{prooftree}
\hfill$

\medskip

$\hfill
\begin{prooftree}
  \hypo{\Gamma \vdash t \cln A}
  \infer1[($z$ not free in $\Gamma$)]{\Gamma \vdash t \cln \forall z.\, A}
\end{prooftree}
\hfill
\hfill
\begin{prooftree}
  \hypo{\Gamma \vdash t \cln \forall z.\, A}
  \infer1{\Gamma \vdash t \cln A[z := b]}
\end{prooftree}
\hfill$

\medskip

$\hfill
\begin{prooftree}
  \infer0{\Gamma, x \cln A \vdash x \cln A}
\end{prooftree}
\hfill
\hfill
\begin{prooftree}
  \infer0{\Gamma \vdash \cc \cln ((A \limp B) \limp A) \limp A}
\end{prooftree}
\hfill$

\medskip

$\hfill
\begin{prooftree}
  \hypo{\Gamma[z := a] \vdash t \cln a \neq b }
  \infer1{\Gamma[z := b] \vdash t \cln a \neq b}
\end{prooftree}
\hfill
\hfill
\begin{prooftree}
  \hypo{\Gamma \vdash t \cln a \neq a }
  \infer1{\Gamma \vdash t \cln A}
\end{prooftree}
\hfill$

\smallskip

The first five are the usual rules of natural deduction, and the sixth types $\cc$ with Peirce's law (which allows classical deduction). The last two reformulate the usual elimination and introduction rules for equality using the symbol $\neq$ instead:

\begin{lemma} The following rule is admissible:
\[\begin{prooftree}
  \hypo{\Gamma \vdash t \cln A[z := a]}
  \infer[double]1[.]{\Gamma, x \cln a=b \vdash \cc(\lambda k.\, x(kt)) \cln A[z := b] }
\end{prooftree}\]
\end{lemma}

\begin{proof}
If the typing judgement $\Gamma \vdash t \cln A[z := a]$ is valid, then so is the judgement $\Gamma, x \cln a=b, k \cln A[z := a] \limp a \neq b \vdash t \cln A[z := a]$. Then we can use the following derivation:
\[\begin{prooftree}
  \hypo{\Gamma, x \cln a=b, k \cln A[z := a] \limp a \neq b \vdash t \cln A[z := a]}
  \infer1{\vdots}
  \infer[]1{\Gamma, x \cln a=b, k \cln A[z := a] \limp a \neq b \vdash kt \cln  a \neq b}
  \infer1{\Gamma, x \cln a=b, k \cln A[z := b] \limp a \neq b \vdash kt \cln  a \neq b}
  \infer1{\vdots}
  \infer[]1{\Gamma, x \cln a=b, k \cln A[z := b] \limp a \neq b \vdash x(kt) \cln  \bot}
  \infer1{\Gamma, x \cln a=b, k \cln A[z := b] \limp a \neq b \vdash x(kt) \cln  A[z := b]}
  \infer1{\Gamma, x \cln a=b \vdash \lambda k.\, x(kt) \cln (A[z {:=} b] \limp a \neq b) \limp A[z {:=} b]}
  \infer1{\vdots}
  \infer[]1[.]{\Gamma, x \cln a=b \vdash \cc(\lambda k.\, x(kt)) \cln A[z := b] }
\end{prooftree}\]
 
\end{proof}

%

\subsection{Classical realizability}

\paragraph*{Poles}

A \emph{pole} is a set $\pole$ of processes that is \emph{saturated}, i.e. such that for all processes $p,q$, if $p \succ q$ and $q \in \pole$, then $p \in \pole$.

\paragraph*{Falsity values and truth values}

For each pole $\pole$ and each closed first-order formula $A$, we define inductively its \emph{falsity value} $\falsity{A}_\pole$ (which is a set of stacks) and its \emph{truth value} $\truth{A}_\pole$ (which is a set of closed $\lambda_c$-terms) with respect to $\pole$:
\begin{itemize}
\item $\truth{A}_\pole = \{ t;\, \forall \pi \in \falsity{A}_\pole~ t \st \pi \in \pole \}$,
\item $\falsity{a \neq b}_\pole = \left\{ \arraycolsep=3pt \begin{array}{cl}
\emptyset & \text{if $a \neq b$ is true (in the}  \\
          & \text{Boolean algebra $\{0,1\}$),} \\
\{ \text{all stacks} \} & \text{if $a \neq b$ is false.}\\
\end{array}\right.$
\item $\falsity{A \limp B}_\pole = \{ t \sc \pi;\, t \in \truth{A}_\pole, \pi \in \falsity{B}_\pole \}$,
\item $\falsity{\forall z\, A}_\pole = \falsity{A[z := 0]}_\pole \cup \falsity{A[z := 1]}_\pole$.
\end{itemize}
We say that a given closed $\lambda_c$-term $t$ \emph{realizes} a given closed first-order formula $A$ \emph{with respect to} a given pole $\pole$ if $t \in \truth{A}_\pole$.

\paragraph*{Adequacy} A key fact about classical realizability is that it is compatible with the above typing rules, and therefore with classical reasoning:

\begin{lemma}[Adequacy lemma] \label{lemma:adequacy} Let $x_1 \cln A_1\cldotsc x_n \cln A_n \vdash t \cln B$ be a valid typing judgement (with $A_1\cldotsc A_n$ closed), and let $u_1\cldotsc u_n$ be closed $\lambda_c$-terms. For all poles $\pole$, if $u_1\cldotsc u_n$ realize $A_1\cldotsc A_n$ respectively with respect to $\pole$, then $t[\multi{x} := \multi{u}]$ realizes $B$ with respect to $\pole$.
\end{lemma}

\subsection{Realizability theories}

We would like to associate with each pole a first-order theory of ``all first-order formulas that are realized with respect to that pole''. However, given any non-empty pole $\pole$, there is bound to be a closed $\lambda_c$-term which realizes $\bot$ (namely: take any $t \st \pi \in \pole$ and consider $tk_\pi$). Therefore, in order to obtain a meaningful notion, we must put some restrictions on which terms are allowed as realizers.

We call \emph{proof-like} any closed $\lambda_c$-term which contains no stack constants ($k_\pi$) and no restricted instructions ($\instrR_n$).

\begin{definition} Let $\pole$ be a pole. The \emph{first-order theory of} $\pole$ is the set of all closed first-order formulas which are realized \emph{by at least one proof-like term} with respect to $\pole$. We denote it by $\th\pole$.
\end{definition}

\begin{remark}
As stated in the introduction, the benefit of restricting the formulas to the language of Boolean algebra is that each pole $\pole$ can be interpreted as a realizability model whose characteristic Boolean algebra's theory is simply $\th\pole$: no translation is needed.
\end{remark}

We say that a pole $\pole$ is \emph{consistent} if its first-order theory is, i.e. if there exists a first-order structure which satisfies $\th\pole$.

\paragraph*{Logical closure} As a consequence of the adequacy lemma and the completeness theorem of first-order logic, the first-order theory of a pole is always \emph{closed under classical deduction}:

\begin{lemma}\label{lemma:logical-closure} Let $\pole$ be a pole and $A$ a closed first-order formula. If $\th\pole$ implies $A$ (in the sense that any first-order structure which satisfies $\th\pole$ also satisfies $A$), then $A \in \th\pole$.
\end{lemma}

In particular, $\pole$ is consistent if and only if no proof-like term realizes $\bot$ with respect to it.

\begin{remark}
In fact, because we chose a very restricted language for formulas, Lemma \ref{lemma:logical-closure} would hold even in the absence of the instruction $\cc$. However, the goal here is to describe a method which can be generalised to richer contexts, and that requires an instruction such as $\cc$ that is capable of altering the control flow: otherwise, all we get is closure under intuitionistic deduction.
\end{remark}

\paragraph*{Boolean algebras} In addition, the first-order theory of a pole is always an extension of $\mathcal{T}_{\operatorname{Bool}}$, and therefore any first-order structure which satisfies it is a Boolean algebra with at least two elements. This is a consequence of the following lemma:

\begin{lemma}\label{lemma:eq-diff-real} Let $A$ be a closed first-order formula that is true in the Boolean algebra $\{0,1\}$.
\begin{itemize}
\item If $A$ is of the form $\forall \multi{z}~ a \neq b$, then $A$ is realized by all closed $\lambda_c$-terms, \emph{universally} (i.e. with respect to all poles);
\item If $A$ is of the form $\forall \multi{z}~ a = b$, then $A$ is universally realized by $\lambda x.\, x$.
\end{itemize}
\end{lemma}

\begin{corollary} \label{cor:realize-ba} For all poles $\pole$, $\th{\pole}$ contains $\mathcal{T}_{\operatorname{Bool}}$.
\end{corollary}

\begin{proof}[Proof of Lemma \ref{lemma:eq-diff-real}]
Let $\pole$ be a pole. First part: for all lists $\multi{w}$ of elements of $\{0,1\}$, we know that $a \neq b$ is true in $\{0,1\}$, therefore the falsity value of $\forall \multi{z}~ a \neq b$ is empty, and so this first-order formula is realized by all closed $\lambda_c$-terms.

Second part: let $\multi{\alpha}$ be a list of elements of $\{0,1\}$. We must prove that $\lambda x.\, x$ realizes $(a = b)[\multi{z} := \multi{\alpha}]$. Let $t\sc\pi$ be in the falsity value of $(a = b)[\multi{z} := \multi{\alpha}]$, i.e. $t$ realizes $a \neq b$ and $\pi$ is any stack. Since $a \neq b$ is false, its falsity value contains all stacks, which means that $t \st \pi$ is in $\pole$. Since $\lambda x.\, x \st t \sc \pi$ evaluates to $t \st \pi$, it is also in $\pole$.
\end{proof}

\paragraph*{Horn clauses}

We have just seen that any (universally quantified) equation or non-equation that is true in the Boolean algebra $\{0,1\}$ is universally realized. In fact, this ``transfer'' property holds for all Horn clauses:

A \emph{Horn clause} is a closed first-order formula of the form either $\forall \multi{z}~ (a_1 = a_1' \limp \ldots \limp a_n = a_n' \limp b = b')$ (\emph{definite clause}) or $\forall \multi{z}~ (a_1 = a_1' \limp \ldots \limp a_n = a_n' \limp b \neq b')$ (\emph{goal clause}).

\begin{lemma}\label{lemma:horn-ba}
Let $A$ be a Horn clause. Then $A$ is true in the Boolean algebra $\{0,1\}$ if and only if it is true in all Boolean algebras with at least two elements.
\end{lemma}

\begin{corollary} Let $A$ be a Horn clause. If $A$ is true in the Boolean algebra $\{0,1\}$, then it is universally realized.
\end{corollary}

\begin{proof}[Proof of Lemma \ref{lemma:horn-ba}] Assume that $A$ is true in $\{0,1\}$, and let $\mathbb B$ be a Boolean algebra with at least two element. There exists a nonempty set $X$ and an injective morphism of Boolean algebras (i.e. a non-necessarily elementary embedding) $\varphi$ from $\mathbb B$ to $\{0,1\}^X$ \cite[Corollary IV.1.12]{bss:universal-algebra}(alternatively, this is also an immediate consequence of Stone's representation theorem). Since $A$ is a universal ($\Pi^0_1$), it is sufficient to prove that $A$ is true in the Boolean algebra $\{0,1\}^X$.

Let $\multi{\delta} \in \{0,1\}^X$. Let $\alpha_1, \alpha_1'\cldotsc \alpha_n, \alpha_n', \beta, \beta'$ be respectively the values of $a_1[\multi{z} := \multi{\delta}]\cldotsc b'[\multi{z} := \multi{\delta}]$ in $\{0,1\}^X$.

Assume that for all $i \leq n$, $\alpha_i = \alpha_i'$, i.e. for all $x \in X$, $\alpha_i(x) = \alpha_i'(x)$.

For all $x \in X$, evaluation at $x$ is a morphism of Boolean algebras from $\{0,1\}^X$ to $\{0,1\}$. Therefore, the value of $a_1[\multi{z} := \multi{\delta}(x)]$ in $\{0,1\}$ is $\alpha_i(x)$, the value of $a_1'[\multi{z} := \multi{\delta}(x)]$ in $\{0,1\}$ is $\alpha_i'(x)$, etc.

Therefore, if $A$ is definite, then for all $x$, we have $\beta(x) = \beta'(x)$, because $A$ is true in $\{0,1\}$. In other words, $\beta = \beta'$, which means that $A$ is true in $\{0,1\}^X$.

On the other hand, if $A$ is a goal clause, then for all $x$, we have $\beta(x) \neq \beta'(x)$. Since $X$ is non-empty, this means that $\beta \neq \beta'$, and so $A$ is true in $\{0,1\}^X$.
\end{proof}

With all these conventions written down, we can precisely state the main result of this paper:

\begin{theorem}\label{thm:main} Let $\mathcal T$ be a first-order theory. The following two statements are equivalent:
\begin{itemize}
\item $\mathcal T$ is closed under classical deduction and contains the theory of Boolean algebras with at least two elements;
\item There exists a pole whose theory is exactly $\mathcal T$.
\end{itemize}
\end{theorem}

In particular, a first-order formula is universally realized if and only if it is true in every Boolean algebra with at least two elements.

We have already seen that that the second point implies the first. The task of the remainder of this paper will be to prove the converse implication.

\section{Constructing the pole}\label{section:construction}

From now on, $\mathcal T$ will denote a fixed first-order theory which is closed under classical deduction and contains $\mathcal T_{\operatorname{Bool}}$. We will construct a pole $\pole_{\mathcal T}$ whose theory is exactly $\mathcal T$.

For each first-order formula $A$ (closed or not), let $\gamma_A$ denote:
\begin{itemize}
\item one of the unrestricted instructions if $A \in \mathcal T$,
\item one of the restricted instructions otherwise.
\end{itemize}
Furthermore, let the $\gamma_A$ be pairwise distinct.

We will construct a pole $\pole_{\mathcal T}$ in such a way that $\gamma_A$ realizes $A$ for all closed first-order formulas $A$: this will imply that its theory contains $\mathcal T$. Then, we will prove the converse inclusion.

\subsection{The structure of first-order formulas}
\label{subsec:structure-formulas}

Any first-order formula $A$ can be decomposed as:
\[
\forall \multi{y}_1~ (B_1 \limp \ldots \limp \forall \multi{y}_{m}~ (B_{m} \limp \forall \multi{y}_{m+1}~ b \neq b') \ldots ),
\]
with $n \geq 0$, $B_1\cldotsc B_{m}$ first-order formulas, each $\multi{y}_{i}$ a list of variables, and $b,b'$ two first-order terms. Moreover, this decomposition is unique, up to renaming of the variables $\multi{y}_{1}\cldotsc \multi{y}_{m+1}$.

Each $B_i$ can itself be decomposed as:
\[
\forall \multi{z}_{i,1}\, (C_{i,1} {\limp} \ldots {\limp}  \forall \multi{z}_{i,n_i}\, (C_{i,n_i} {\limp} \forall \multi{z}_{i,n_i+1}\, c_i \neq c'_i) \ldots ),
\]
so that $A$ is decomposed as:

~\\
\begin{tikzpicture}

\node at (0,0) {$\forall \multi{y}_1~ (\dashedph \limp \ldots \limp \forall \multi{y}_{m}~ (\dashedph \limp \forall \multi{y}_{m+1}~ b \neq b') \ldots )$};

\draw [->, decorate, line width=0.7pt] (-2.49,-0.6) -- (-2.49,0.06);
\draw [decorate, line width=1pt, decoration = {brace, raise=8pt, aspect=0.203}] (-4.2,-1) --  (4.2,-1);
\node at (0,-1) {$\forall \multi{z}_{1,1}\, (C_{1,1} {\limp} \ldots {\limp}  \forall \multi{z}_{1,n_1}\, (C_{1,n_1} {\limp} \forall \multi{z}_{1,n_1+1}\, c_1 \neq c'_1) \ldots )$};

\draw [->, decorate, line width=0.7pt] (0.13,-1.6) -- (0.13,0.06);
\draw [decorate, line width=1pt, decoration = {brace, raise=8pt, aspect=0.515}] (-4.4,-2) --  (4.4,-2);
\draw [->] (-2.3,-2) -- (-2.1,-2);
\draw [->] (-1.65,-2) -- (-1.45,-2);
\draw [->] (0.45,-2) -- (0.65,-2);
\node at (0,-2) {$\forall \multi{z}_{m,1} (C_{m,1} ~\, \ldots ~\, \forall \multi{z}_{m,n_{m}} (C_{m,n_{m}} ~\,\forall \multi{z}_{m,n_{m}+1} c_{m} \neq c'_{m}) \ldots )$};
\end{tikzpicture}

Whenever we decompose a formula $A$ in this way, we will denote by $\multi{\multi{y}}$ the concatenated list $\allowbreak \multi{y}_1\cldotsc \multi{y}_{m+1}$, and for all $1 \leq i \leq m$, we will denote by $\multi{\multi{z}}_i$ the concatenated list $\multi{z}_{i,1}\cldotsc \multi{z}_{i,n_{i}+1}$. Furthermore, we will assume that the variables $\allowbreak \multi{\multi{y}},\allowbreak \multi{\multi{z}}_1\cldotsc \multi{\multi{z}}_m$ are chosen all different from one another and from the free variables of $A$.

\subsection{The pole $\pole_{\mathcal T}$}
\label{subsec:definition-pole}

We define by induction an increasing sequence $(\pole_{\mathcal T, k})_{k\in\mathbb N}$ of sets of processes: $\pole_{\mathcal T, 0}$ is empty, and for all $k$, $\pole_{\mathcal T, k+1}$ is the smallest set of processes such that:
\begin{itemize}
\item For all processes $p,q$, if $p \succ_1 q$ and $q \in \pole_{\mathcal T, k}$, then $p \in \pole_{\mathcal T, k+1}$;
\item For each closed first-order formula $A$ (decomposed as in section \ref{subsec:structure-formulas}), for all closed $\lambda_c$-terms $t_1, \ldots, t_m$, all stacks $\pi$ and all lists $\multi{\multi{\beta}} \in \{0,1\}$ such that $(b = b')\left[\multi{\multi{y}} := \multi{\multi{\beta}}\right]$ is true, if the set of processes
\[\arraycolsep=2pt\left\{\begin{array}{c}
t_i \st \gamma_{C_{i,1}\left[\multi{\multi{z_i}} := \multi{\multi{\delta_i}}, \multi{\multi{y}} := \multi{\multi{\beta}}\right]} \sc \ldots \sc \gamma_{C_{i,n_i}\left[\multi{\multi{z_i}} := \multi{\multi{\delta_i}}, \multi{\multi{y}} := \multi{\multi{\beta}}\right] } \sc \pi;\\

i \leq m \text{ and } \multi{\multi{\delta_i}} \in \{0,1\} \text{ such that}\\

(c_i = c'_i)\left[\multi{\multi{z_i}} := \multi{\multi{\delta_i}}, \multi{\multi{y}} := \multi{\multi{\beta}}\right] \text{ is true}
\end{array}\right\}\]
is included in $\pole_{\mathcal T, k}$, then the process
\[
\gamma_A \st t_1 \sc \ldots \sc t_m \sc \pi
\]
is in $\pole_{\mathcal T, k+1}$.
\end{itemize}

Then, we define the pole $\pole_{\mathcal T}$ as the directed union\\ $\bigcup_{k\in\mathbb N} \pole_{\mathcal T,k}$.

\begin{remark} The rule for the instructions $\gamma_A$ have the following general shape:
\[ \arraycolsep=0pt \begin{array}{l}
\text{If there exists } \multi{\multi \beta} \text{ such that for all } i,\multi{\multi \delta}_i \text{, }
	t_i \st \pi'_{\multi{\multi \beta}, i, \multi{\multi \delta}_i} \text{ is in } \pole_{\mathcal T} \text{,}\\ 
\text{then } \gamma_A \st t_1 \sc \ldots \sc t_m \sc \pi \text{ is in } \pole_{\mathcal T} \text{.}
\end{array} \]
It has been pointed out \cite{geoffroy:nondeterminism} that such a rule can be interpreted as saying that $\gamma_A$ is a special kind of nondeterministic instruction: part ``may'' (because of the existential quantification), and part ``must'' (because of the universal quantification).
\end{remark}

\section{The theory of $\pole_{\mathcal T}$ contains $\mathcal T$}\label{section:contains-T}

We wish to prove that $\th{\pole_{\mathcal T}}$ contains $\mathcal T$. Since $\gamma_A$ is proof-like whenever $A$ is in $\mathcal T$, it is sufficient to prove the following result:

\begin{proposition} For all closed first-order formulas $A$, $\gamma_A$ realizes $A$ with respect to $\pole_{\mathcal T}$.
\end{proposition}
\begin{proof}
We proceed by induction on the height of $A$. Let $A$ be decomposed as in $\ref{subsec:structure-formulas}$.

Let $t_1 \sc \ldots \sc t_n \sc \pi$ be in the falsity value of $A$. In other words, let $\multi{\multi{\beta}}$ be a list of elements of $\{0,1\}$, let $t_1, \ldots t_n$ be closed $\lambda_c$-terms such that $t_i$ realizes $B_i[\multi{\multi{y}} := \multi{\multi{\beta}}]$ for all $i$, and let $\pi$ be in  the falsity value of $(b \neq b')[\multi{\multi{y}} := \multi{\multi{\beta}}]$, which is the same as saying that $(b = b')[\multi{\multi{y}} := \multi{\multi{\beta}}]$ is true. All we need to do is prove that
\[\gamma_A \st t_1 \sc \ldots \sc t_m \sc \pi\]
is in $\pole_{\mathcal T}$.

Let $i \leq m$ and $\multi{\multi{\delta_i}} \in \{0,1\}$ be such that $(c_i = c'_i)[\multi{\multi{z_i}} := \multi{\multi{\delta_i}}.,\allowbreak. \multi{\multi{y}} := \multi{\multi{\beta}}]$ is true. By the induction hypothesis, we know that for all $j \leq n_i$, $\gamma_{C_{i,j}[\multi{\multi{z_i}} := \multi{\multi{\delta_i}}, \multi{\multi{y}} := \multi{\multi{\beta}}]}$ realizes $C_{i,j}[\multi{\multi{z_i}} := \multi{\multi{\delta_i}}, \multi{\multi{y}} := \multi{\multi{\beta}}]$. Since $t_i$ realizes $B_i[\multi{\multi{y}} := \multi{\multi{\beta}}]$ and $\pi$ is in the falsity value of $(c_i \neq c'_i)[\multi{\multi{z_i}} := \multi{\multi{\delta_i}}, \multi{\multi{y}} := \multi{\multi{\beta}}]$ (because this inequality is false), we have that
\[t_i \st \gamma_{C_{i,1}[\multi{\multi{z_i}} := \multi{\multi{\delta_i}}, \multi{\multi{y}} := \multi{\multi{\beta}}]} \sc \ldots \sc \gamma_{C_{i,n_i}[\multi{\multi{z_i}} := \multi{\multi{\delta_i}}, \multi{\multi{y}} := \multi{\multi{\beta}}] } \sc \pi\]
is in $\pole_{\mathcal T}$.

Therefore, by definition of $\pole_{\mathcal T}$, $\gamma_A \st t_1 \sc \ldots \sc t_m \sc \pi$ is in $\pole_{\mathcal T}$.

\end{proof}

\section{The theory of $\pole_{\mathcal T}$ is contained in $\mathcal T$}\label{section:contained-in-T}

All that remains is to prove that The theory $\th{\pole_{\mathcal T}}$ is contained in $\mathcal T$.

For all $\lambda_c$-terms $t$ (respectively, all stacks $\pi$; all processes $p$), let $\mathcal C_t$ (respectively, $\mathcal C_\pi$; $\mathcal C_p$) denote the conjunction of all first-order formulas $A$ such that $\gamma_A$ appears anywhere in $t$ (respectively, in $\pi$; in $p$)--including nested within a stack constant $k_\pi$. Note that $\mathcal C_t$ is not necessarily closed even if $t$ is (because the former notion of closure is about first-order variables, while the latter is about variables of the $\lambda_c$-calculus).

We are going to prove that for all processes $p$ such that $\mathcal C_p$ is closed, if $p$ is in $\pole_{\mathcal T}$, then $p$ must contain a contradiction, in the sense that $\mathcal C_p$ must be false in all Boolean algebras with at least two elements.

In order to prove this, we will need to state and prove a more general result that also covers the case when $\mathcal C_p$ is not closed. To that end, we will need the following notation: for all $\lambda_c$-terms $t$, all lists $\multi{z}$ of first-order variables, and all lists $\multi{b}$ of first-order terms, we denote by $t[\multi{z} := \multi{b}]$ the $\lambda_c$-term obtained by replacing each instruction of the form $\gamma_A$ by $\gamma_{A[\multi{z} := \multi{b}]}$ (including when they appear nested within a stack constant). Similarly, we define $\pi[\multi{z} := \multi{b}]$ when $\pi$ is a stack, and $p[\multi{z} := \multi{b}]$ when $p$ is a process.

\begin{proposition} \label{prop:realizes-implies}
Let $p$ be a process, $\multi{a} = a_1, \ldots, a_r$ and $\multi{a}' = a_1', \ldots, a_r'$ two lists of first-order terms, and $\multi{w}$ a list of distinct first-order variables that contains all the free variables of $\mathcal C_p$, $\multi{a}$ and $\multi{a}'$.

Assume that for all lists $\multi{\alpha}$ of elements of $\{0,1\}$ such that $(\multi{a} = \multi{a}')[\multi{w} := \multi{\alpha}]$ is true, $p[\multi{w} := \multi{\alpha}]$ is in $\pole_{\mathcal T}$ (where $\multi{a} = \multi{a}'$ denotes the conjunction $(a_1 = a_1') \wedge \ldots \wedge (a_r = a_r')$).

Then the first-order formula $\exists \multi{w}~ (\mathcal C_p \wedge (\multi{a} = \multi{a}'))$ is false in all Boolean algebras with at least two elements .
\end{proposition}

\begin{corollary}
Let $t$ be a closed $\lambda_c$-term such that $\mathcal C_t$ is closed, and $A$ a closed first-order formula. If $t$ realizes $A$ with respect to $\pole_{\mathcal T}$, then the formula $\mathcal C_t \limp A$ is true in all Boolean algebras with at least two elements.
\end{corollary}

\begin{proof}
If $t$ realizes $A$ with respect to $\pole_{\mathcal T}$, then $\gamma_{A \limp \bot} \st t \sc \omega$ is in $\pole_{\mathcal T}$, therefore $\mathcal C_t \wedge (A \limp \bot)$ is false in all Boolean algebras with at least two elements.
\end{proof}

\begin{corollary}
The theory $\th{\pole_{\mathcal T}}$ is contained in $\mathcal T$.
\end{corollary}

\begin{proof}
Let $A \in \th{\pole_{\mathcal T}}$. Let $t$ be a proof-like term which realizes $A$. The formula $\mathcal C_t$ is in $\mathcal T$ by construction, because $t$ is proof-like. By the previous corollary, the formula $\mathcal C_t \limp A$ is also in $\mathcal T$, therefore $A$ is in $\mathcal T$.
\end{proof}

We now prove the proposition:

\begin{proof}[Proof of Proposition \ref{prop:realizes-implies}]
We will prove by induction that for all natural numbers $k$, for all $p$, $\multi{a}$, $\multi{a}'$, $\multi{w}$, if $p[\multi{w} := \multi{\alpha}]$ is in $\pole_{\mathcal T,k}$ for all $\multi{\alpha} \in \{0,1\}$, then the first-order formula $\exists \multi{w}~ (\mathcal C_p \wedge \multi{a} = \multi{a}')$ is false in all Boolean algebras with at least two elements. (This is sufficient to prove the proposition because the sequence $(\pole_{\mathcal T,k})_{k\in\mathbb N}$ is cumulative and the set of all $\multi{\alpha} \in \{0,1\}$ is finite.)

Note that if the formula $\exists \multi{w}~ (\multi{a} = \multi{a}')$ is false in $\{0,1\}$, then it is false in all Boolean algebras, because its negation is equivalent to a Horn clause (Lemma \ref{lemma:horn-ba}). In particular, $\exists \multi{w}~ (\mathcal C_p \wedge \multi{a} = \multi{a}')$ is false in all Boolean algebras with at least two elements. Therefore, from now on, we will assume that $\exists \multi{w}~ (\multi{a} = \multi{a}')$ is true in $\{0,1\}$.

The result is vacuously true for $k=0$, because $\pole_{\mathcal T,0}$ is empty.

Assume the result holds for some $k$, and let $p$, $\multi{a}$, $\multi{a}'$, $\multi{w}$ be such that $p[\multi{w} := \multi{\alpha}]$ is in $\pole_{\mathcal T,k+1}$ for all $\multi{\alpha} \in \{0,1\}$ such that $(\multi{a} = \multi{a}')[\multi{w} := \multi{\alpha}]$ is true.

If we look back at the definition of $\pole_{\mathcal T}$ in section \ref{subsec:definition-pole}, we see that we must be in one of the following cases:

(i) There exists a process $q$ such that $p$ evaluates in one step to $q$. In that case, for all $\multi{\alpha} \in \{0,1\}$, if $(\multi{a} = \multi{a}')[\multi{w} := \multi{\alpha}]$ is true, then $q[\multi{w} := \multi{\alpha}]$ must be in $\pole_{\mathcal T,k}$. Therefore, by the induction hypothesis, the formula $\exists \multi{w}~ (\mathcal C_q \wedge \multi{a} = \multi{a}')$ is false in all Boolean algebras with at least two elements. On the other hand, evaluation can only remove or copy the constants $\gamma_A$, and not add new ones. This means that the formula $\forall \multi{w}~ (\mathcal C_p \limp \mathcal C_q)$ is a propositional tautology, which proves the result.

(ii) The process $p$ is of the form $\gamma_A \st t_1 \sc \ldots \sc t_n \sc \pi$. In that case, let $A$ be decomposed as in section \ref{subsec:structure-formulas}. Then for all $\multi{\alpha}$ in $\{0,1\}$ such that $(\multi{a} = \multi{a}')[\multi{w} := \multi{\alpha}]$ is true, there exists a list $\multi{\multi{\beta}}_{\multi{\alpha}}$ in $\{0,1\}$ such that $(b = b')\left[\multi{w} := \multi{\alpha}, \multi{\multi{y}} := \multi{\multi{\beta}}_{\multi{\alpha}}\right]$ is true and that set of processes
\[\arraycolsep=2pt\left\{\begin{array}{c}
t_i[\multi{w} := \multi{\alpha}] \st \gamma_{C_{i,1}\left[\multi{w} := \multi{\alpha},\multi{\multi{z_i}} := \multi{\multi{\delta_i}}, \multi{\multi{y}} := \multi{\multi{\beta}}_{\multi{\alpha}}\right]} \sc \ldots \sc \pi[\multi{w} := \multi{\alpha}];\\

i \leq m \text{ and } \multi{\multi{\delta_i}} \in \{0,1\} \text{ such that}\\

(c_i = c'_i)\left[\multi{w} := \multi{\alpha}, \multi{\multi{z_i}} := \multi{\multi{\delta_i}}, \multi{\multi{y}} := \multi{\multi{\beta}}_{\multi{\alpha}}\right] \text{ is true}
\end{array}\right\}\]
is included in $\pole_{\mathcal T, k}$.

Every function from $\{0,1\}^k$ to $\{0,1\}$ can be represented by a first-order term with $k$ free variables. Therefore, one can choose a list $\multi{\multi{e}}$ of first-order terms with no free variables other than $\multi{w}$ such that for all $\multi{\alpha}$ in $\{0,1\}$, if $(\multi{a} = \multi{a}')[\multi{w} := \multi{\alpha}]$ is true, then the value of $\multi{\multi{e}}[\multi{w} := \multi{\alpha}]$ (in $\{0,1\}$) is $\multi{\multi{\beta}}_{\multi{\alpha}}$.

Let $i \leq m$. For all $\multi{\alpha},\multi{\multi{\delta_i}}$ in $\{0,1\}$ such that $((\multi{a} = \multi{a}') \wedge (c_i = c_i')[\multi{\multi{y}} := \multi{\multi{e}}])[\multi{w} := \multi{\alpha}, \multi{\multi{z_i}} := \multi{\multi{\delta_i}}]$ is true, we know that the process
\[
(t_i \st \gamma_{C_{i,1}[\multi{\multi{y}} := \multi{\multi{e}}]} \sc \ldots \sc \gamma_{C_{i,n_i}[\multi{\multi{y}} := \multi{\multi{e}}]} \sc \pi)[\multi{w} := \multi{\alpha}, \multi{\multi{z_i}} := \multi{\multi{\delta_i}}]
\]
is in $\pole_{\mathcal T, k}$. Therefore, by the induction hypothesis, we know that the formula
\begin{multline*}
\exists \multi{w}~ \exists \multi{\multi{z}}_i~ (\mathcal C_{t_i} \wedge C_{i,1}[\multi{\multi{y}} := \multi{\multi{e}}] \wedge \ldots \wedge C_{i,n_i}[\multi{\multi{y}} := \multi{\multi{e}}] \\ \wedge \mathcal C_\pi \wedge (\multi{a} = \multi{a}') \wedge (c_i = c_i')[\multi{\multi{y}} := \multi{\multi{e}}])
\end{multline*}
is false in all Boolean algebras with at least two elements. In other words, for all $i \leq m$, the formula
\[
\exists \multi{w}~ (\mathcal C_{t_i} \wedge \mathcal C_\pi \wedge (\multi{a} = \multi{a}') \wedge (B_i[\multi{\multi{y}} := \multi{\multi{e}}] \limp \bot))
\]
is false in all Boolean algebras with at least two elements.

This means that the formula
\begin{multline*}
\exists \multi{w}~ (\mathcal C_{t_1} \wedge \ldots \wedge C_{t_m} \wedge (\multi{a} = \multi{a}') \wedge \mathcal C_\pi \\
\wedge (B_1 \limp \ldots \limp B_m \limp \bot) [\multi{\multi{y}} := \multi{\multi{e}}])
\end{multline*}
is false in all Boolean algebras with at least two elements.

The formula $\forall \multi{w}~ (\multi{a} = \multi{a}' \limp (b = b')[\multi{\multi{y}} := \multi{\multi{e}}])$ is true in $\{0,1\}$. Since it is a Horn clause, it is true in all Boolean algebras with at least two elements (Lemma \ref{lemma:horn-ba}). Therefore the formula 
\begin{multline*}
\exists \multi{w}~ (\mathcal C_{t_1} \wedge \ldots \wedge C_{t_m} \wedge (\multi{a} = \multi{a}') \wedge \mathcal C_\pi \\
\wedge (B_1 \limp \ldots \limp B_m \limp b \neq b') [\multi{\multi{y}} := \multi{\multi{e}}])
\end{multline*}
is false in all Boolean algebras with at least two elements. This formula is a logical consequence of the formula
\[
\exists \multi{w}~ (\mathcal C_{t_1} \wedge \ldots \wedge C_{t_m} \wedge \mathcal C_\pi \wedge (\multi{a} = \multi{a}') \wedge A),
\]
which is itself a logical consequence of the formula
\[
\exists \multi{w}~ (\mathcal C_p \wedge (\multi{a} = \multi{a}')).
\]
Therefore, this last formula is false in all Boolean algebras with at least two elements.
\end{proof}

This completes the proof of Theorem \ref{thm:main}.

\section{Application: sequentialisation in a denotational model of the $\lambda_c$-calculus} \label{section:sequential}

It is known \cite{streicher-reus:98} that, by performing Scott's construction $D_\infty$ with $D_0 = \{ \bot, \top \}$ (the two-elements lattice), one obtains a denotational model of the $\lambda_c$-calculus. As with any such model, one natural question \cite{scott:owhy} is: among its elements, which ones are sequentialisable. In other words, which ones are the denotation of an actual $\lambda_c$-term. We will show how the techniques developed in this paper can give a partial answer.

\subsection{The construction of $D_\infty$}

We recall the construction of $D_\infty$ \cite{barendregt:lambda}. The first step is to define a finite lattice $D_n$ for all natural numbers $n$. We let:
\begin{itemize}
\item $D_0 = \{ \bot, \top \}$ (the two-elements lattice, with $\bot < \top$),
\item $D_{n+1} = [D_n \to D_n]$ (the complete lattice of all Scott-continuous functions from $D_n$ to $D_n$, which is the same as the lattice of all non-decreasing functions, because $D_n$ is finite).
\end{itemize}
Then for all $n$ we define an injection $\varphi_n \in [D_n \to D_{n+1}]$ and a projection $\psi_n \in [D_{n+1} \to D_{n}]$:
\begin{itemize}
\item for all $\alpha \in D_0$, $\varphi_0(\alpha)$ is the function $\beta \mapsto \alpha$,
\item for all $f \in D_1$, $\psi_0(f) = f(\bot)$
\item for all $n \geq 0$ and all $\alpha \in D_{n+1}$, $\varphi_{n+1}(\alpha) = \varphi_n \circ \alpha \circ \psi_n$,
\item for all $n \geq 0$ and all $f \in D_{n+2}$, $\psi_{n+1}(f) = \psi_n \circ f \circ \varphi_n$.
\end{itemize}
Finally, we define $D_\infty$ as the limit of the diagram $(D_n, \psi_n)_{n \in \mathbb N}$ in the category of complete lattices and Scott-continuous functions, namely:
\[ D_\infty =\{ \alpha = (\alpha_{[n]} \in D_n)_{n \in \mathbb N};~ \forall n~ \alpha_{[n]} = \psi_n(\alpha_{[n+1]}) \}. \]
In fact, $D_\infty$ is also a colimit of the diagram $(D_n, \varphi_n)_{n \in \mathbb N}$ \cite{barendregt:lambda}; for all $n$, the injection from $D_n$ into $D_{\infty}$ is given by:
\[ \begin{array}{lll}
\alpha_{[n]} & \mapsto & (\psi_0 \circ \ldots \circ \psi_{n-1} (\alpha_{[n]}), \ldots, \psi_{n-1}(\alpha_{[n]})), \\
& & \alpha_{[n]}, \\
& & \varphi_n({\alpha_{[n]}}), \varphi_{n+1}\circ\varphi_{n}(\alpha_{[n]}), \ldots).
\end{array} \]
As is customary with colimits, we identify each $D_n$ with the corresponding subset of $D_\infty$.

This defines an extensional reflexive object in the category of complete lattices and Scott-continuous functions, because we can define two inverse isomorphisms $\Phi : D_\infty \to [D_\infty \to D_\infty]$ and $\Psi : [D_\infty \to D_\infty] \to D_\infty$:
\[ \begin{array}{rcl}
\Phi((\alpha_{[n]})_{n \in \mathbb N}) & = & (\beta{[n]})_{n \in \mathbb N} \mapsto (\alpha_{[n+1]}(\beta_{[n]}))_{n \in \mathbb N}\\
\Psi(f) & = & (\gamma_{[n]})_{n\in\mathbb N}\text{, where} \\
&& \gamma_{[0]} = f(\bot)_{[0]} \in D_0, \\
&& \gamma_{[n+1]} = (\alpha_{[n]} \mapsto f(\alpha_{[n]})_{[n]}) \in D_{n+1}.
\end{array} \]

\paragraph*{A model of the $\lambda_c$-calculus}

It is known \cite{streicher-reus:98} that $D_\infty$ can be equipped with the structure of a model of the $\lambda_c$-calculus. One way to present this structure is as follows:
for each $\lambda_c$-term $t$ and each list $x_1, \ldots, x_n$ of pairwise distinct variables containing at least all the free variables of $t$, define a Scott-continuous function $\ds t : D_\infty^n \to D_\infty$:
\[ \arraycolsep=2pt \begin{array}{rll}
\ds{x_k}(\alpha_1, \ldots, \alpha_n) & = & \alpha_k \\
\ds{\lambda x_{n+1}.\, t}(\alpha_1, \ldots, \alpha_n) & = & \Psi(\alpha_{n+1} \mapsto \ds{t}(\alpha_1, \ldots, \alpha_{n+1})) \\
\ds{t u}(\multi \alpha) & = & \Phi(\ds{t}(\multi \alpha))(\ds{u}(\multi \alpha)) \\
\ds{\instrU_m}(\multi \alpha) = \ds{\instrR_m}(\multi \alpha) & = & \bot\\
\ds{k_{t_1\cdot\ldots\cdot t_m\cdot\omega}}(\multi \alpha) &=& \ds{\lambda y.\, y\,t_1\ldots t_n}(\multi \alpha)\\
\ds{\cc}(\multi \alpha) & = & \bigvee_{\beta,\gamma \in D_{\infty}} ((\beta \to \gamma) \to \beta) \to \beta, \\
&& \text{where }\delta \to \varepsilon \text{ is the least}\\
&& \text{element of } D_\infty\text{such that}\\
&& \Phi(\delta \to \varepsilon)(\delta) \geq \varepsilon. 
\end{array} \]

This structure is compatible with evaluation in the Krivine abstract machine \cite{streicher-reus:98}: for all closed $\lambda_c$-terms $t,t',u_1\ldots, u_n,\allowbreak u'_1\ldots, u'_{n'}$, if \[ t \st u_1 \cdot \ldots u_n \cdot \omega \succ  t' \st u'_1 \cdot \ldots u'_{n'} \cdot \omega, \] then $\ds{t\,u_1\ldots u_n} = \ds{t'\,u'_1\ldots u'_{n'}}$.

In addition, this model charaterises solvability \cite{streicher-reus:98}. Namely, for each closed term $t$, we have $\ds t > \bot$ if and only if there exist $k \leq n \in \mathbb N$ such that for each stack $u_1 \cdot \ldots \cdot u_n \cdot \pi$, there exists a stack $\pi'$ such that \[ t \st u_1 \cdot \ldots \cdot u_n \cdot \pi \succ u_k \st \pi'. \]

\subsection{Sequentialisation}

In this context, the problem of sequentialisation can be formulated as follows: given $\alpha \in D_\infty$, is there a closed $\lambda_c$-term $t$ such that $\ds t = \alpha$? We will show how the techniques developed in this paper can answer a simplified version of this problem. Namely, whenever $\alpha$ is in $D_n$ for some finite $n$, we give a necessary and sufficient condition for the existence of a closed $\lambda_c$-term $t$ such that $\ds t \geq \alpha$.

\begin{remark} Alternatively, one might also ask whether there exists a \emph{proof-like} term $t$ such that $\ds t \geq \alpha$. However, due to how $\ds{\instrR_m}$ and $\ds{k_\pi}$ are defined, we have that for each closed $\lambda_c$-term $t$, there exists a proof-like term $t'$ such that $\ds t = \ds{t'}$, so that is in fact the same question.
\end{remark}

\begin{remark} \label{lemma:d-fin-induction} One can prove that the set $\bigcup_{n\in\mathbb N} D_n$ is in fact the least subset of $D_\infty$ that contains $\bot$ and $\top$ and is closed under the binary operations $\vee$ and $\to$ (where $\delta \to \varepsilon$ is defined as the least element of $D_\infty$ such that $\Phi(\delta \to \varepsilon)(\delta) \geq \varepsilon$).
\end{remark}

\paragraph*{Interpreting first-order formulas in $D_\infty$}

For each closed first-order formula $A$, we can define an element $\ds A \in D_\infty$:
\begin{itemize}
\item $\ds{a \neq b} = \left\{ \begin{array}{ll}
\bot & \text{if } a\neq b \text{ is true,}\\
\top & \text{if } a\neq b \text{ is false}
\end{array}\right.$
\item $\ds{A \to B} = \ds A \to \ds B$, 
\item $\ds{\forall x\, A} = \ds A(0) \vee \ds A(1)$
\end{itemize}

In fact, for each $A$, there exists $n$ such that $\ds A$ is in $D_n$. Conversely, thanks to Remark \ref{lemma:d-fin-induction}, for all $n$ and all $\alpha \in D_n$, one can construct inductively a closed formula $\Theta_\alpha$ such that $\ds {\Theta_\alpha} = \alpha$.

\paragraph*{True formulas give sequentialisable elements}

These two translations, from $\lambda_c$-terms and first-order formulas into $D_\infty$, are linked by a variant of the adequacy lemma (Lemma \ref{lemma:adequacy}), which can be proved using the same standard techniques:

\begin{lemma}Let $x_1 \cln A_1\cldotsc x_n \cln A_n \vdash t \cln B$ be a valid typing judgement (with $A_1\cldotsc A_n$ closed), and let $u_1\cldotsc u_n$ be closed $\lambda_c$-terms. If $\ds {u_1}\geq \ds{A_1}\cldotsc \ds {u_n}\geq \ds{A_n}$, then $\ds t(\ds {u_1}\cldotsc \ds {u_n}) \geq \ds B$.
\end{lemma}

In addition, a variant of Lemma \ref{lemma:eq-diff-real} also holds in $D_\infty$ (with essentially the same proof):

\begin{lemma} Let $A$ be a closed first-order formula that is true in the Boolean algebra $\{0,1\}$.
\begin{itemize}
\item If $A$ is of the form $\forall \multi{z}~ a \neq b$, $\ds A = \bot \leq \ds t$ for all closed $\lambda_c$-terms $t$;
\item If $A$ is of the form $\forall \multi{z}~ a = b$, then $\ds A = (\top \to \top) \leq \ds {\lambda x.\, x}$.
\end{itemize}
\end{lemma}

As a result, for each closed first-order formula $A$, if $A$ is true in all Boolean algebras with at least two elements, then there exists a (proof-like) closed $\lambda_c$-term $t$ such that $\ds t \geq \ds A$.

\paragraph*{Sequentialisation gives universal realizers}

Given two closed first-order formulas $A$ and $B$, if $\ds{A} \geq \ds{B}$, then for all poles $\pole$, we have $\falsity{A}_{\pole} \supseteq \falsity{B}_{\pole}$. This can be proved by induction on the pair $(\min(h_A, h_B), \max(h_A, h_B))$ (where $h_A$ and $h_B$ denote the heights of $A$ and $B$ respectively), using the decomposition from section \ref{subsec:structure-formulas} (the single-level version).

In addition, for each closed $\lambda_c$-term $t$ and each closed first-order formula $A$, one can prove by induction on the structure of $t$ that if $\ds t \geq \ds A$, then $t$ realizes $A$ universally. More precisely, for each pole $\pole$ and each $\lambda_c$-term $t$ with free variables $\multi{x}$, one can prove by induction on $t$ that for all closed formulas $A, \multi{B}$, if $\ds{t}(\ds{\multi{B}}) \geq \ds{A}$, then for all $\multi{s}$ that realize $\multi{B}$, $t[\multi{x} := \multi{s}]$ realizes $A$. The proof has much in common with the proof of the adequacy lemma, but it relies heavily on the previous paragraph.

Thanks to Theorem \ref{thm:main} (and Remark \ref{lemma:d-fin-induction}), this means that for all closed first-order formulas $A$, if there exists a closed $\lambda_c$-term $t$ such that $\ds t \geq \ds A$, then $A$ is true in all Boolean algebras with at least two elements.

Combining all these results, we get:

\begin{proposition} Let $n\in\mathbb N$ and $\alpha \in D_n$. The following two statements are equivalent:
\begin{itemize}
\item There exists a closed $\lambda_c$-term $t$ such that $\ds t \geq \alpha$,
\item The formula $\Theta_\alpha$ is true in all Boolean algebras with at least two elements (where $\Theta_\alpha$ is any closed first-order formula such that $\ds {\Theta_\alpha} = \alpha$. Such a $\Theta_\alpha$ does exist for all $\alpha$ and it can be obtained effectively. The choice of $\Theta_\alpha$ does not matter).
\end{itemize}
\end{proposition}

\section{Concluding remarks}

We have proved that the only \emph{first-order} formulas that are true in the characteristic Boolean algebra ($\gimel 2$) of every classical realizability model are those that are true in all Boolean algebras with at least two elements. In a sense, as far as the first order is concerned, the only thing we always know about $\gimel 2$ is that it is a Boolean algebra with at least two elements. This does not extend to the second order: indeed, for example, Krivine \cite{krivine:ra4} has proved that there always exists an ultrafilter on $\gimel 2$, even though the axiom of choice does not necessarily hold in a realizability model. This raises the question: what are the second- and higher-order properties of $\gimel 2$ that are true in all realizability models? And what about $\gimel \mathbb N$?

In a different direction, it would be interesting to know if and to what extent the technique presented in section \ref{section:sequential} can be adapted to other denotational models of the lambda-calculus, and notably to non-lattice and non-continuations-based models.


\balance



\end{document}